\documentclass[11pt]{article}
\usepackage{chao}
\usepackage{graphicx}
\usepackage{caption}
\usepackage{subcaption}

\newcommand{\eps}{\e}

\newcommand{\strength}{\gamma}
\newcommand{\cost}{\zeta}
\newcommand{\component}{\kappa}

\begin{document}

\title{A note on approximate strengths of edges in a hypergraph}
\author{
  Chandra Chekuri\thanks{Department of Computer Science, 
    University of Illinois, Urbana, IL 61801. \textsc{chekuri@illinois.edu}. Work on this paper supported in part by NSF grant CCF-1526799.}
\and
    Chao Xu\thanks{Department of Computer Science, University of Illinois, Urbana, IL 61801. \textsc{chaoxu3@illinois.edu}. Work on this paper
supported in part by NSF grant CCF-1526799.}
}
\date{\today}

\maketitle

\begin{abstract}
  Let $H=(V,E)$ be an edge-weighted hypergraph of rank $r$.  Kogan and
  Krauthgamer~\cite{KoganK15} extended Bencz\'{u}r and Karger's
  \cite{BenczurK15} random sampling scheme for cut sparsification from
  graphs to hypergraphs. The sampling requires an algorithm for
  computing the approximate strengths of edges. In this note we extend
  the algorithm for graphs from \cite{BenczurK15} to hypergraphs and
  describe a near-linear time algorithm to compute approximate
  strengths of edges; we build on a sparsification result for
  hypergraphs from our recent work \cite{ChekuriX17}. Combined with
  prior results we obtain faster algorithms for finding
  $(1+\eps)$-approximate mincuts when the rank of the
  hypergraph is small.
\end{abstract}

\section{Introduction}
\label{sec:intro}
Bencz\'{u}r and Karger, in their seminal work \cite{BenczurK15}, showed
that all cuts of a weighted graph $G=(V,E)$ on $n$ vertices can be
approximated to within a $(1\pm\eps)$-factor by a sparse weighted graph
$G'=(V,E')$ where $|E'| = O(n \log n/\eps^2)$. Moreover, $G'$ can be
computed in near-linear time by a randomized algorithm. The algorithm
has two steps. In the first step, it computes for each edge $e$ a
number $p_e$ which is proportional to the inverse of the approximate strength of edge
$e$. The second step is a simple sampling scheme where each edge is
independently sampled with probability $p_e$ and if it is chosen, the
edge $e$ is assigned a capacity $u_e/p_e$ where $u_e$ is the original
capacity/weight of $e \in E$.  It is shown in \cite{BenczurK15} that
the probabilities $p_e, e \in E$ can be computed by a deterministic
algorithm in $O(m \log^3 n)$ time where $m = |E|$ if $G$ is weighted
and in $O(m \log^2 n)$ time if $G$ is unweighted or has polynomially
bounded weights. More recent work \cite{FungHHP11} has shown 
a general framework for cut sparsification where the sampling probability
$p_e$ can also be chosen to be approximate connectivity between the end
points of $e$. In another seminal work, Batson, Spielman and Srivastava
\cite{BatsonSS12} showed that spectral-sparsifiers, which are stronger
than cut-sparsifiers, with $O(n/\eps^2)$ edges
exist, and can be computed in polynomial time. Lee and Sun recently showed
such sparsifiers can be computed in $\tilde{O}(m)$ time, 
where $\tilde{O}$ hide polylog factors \cite{LeeS17}. 

In this paper, we are interested in hypergraphs. $H=(V,E)$ is
a hypergraph where each $e \in E$ is a subset of nodes. Graphs are a special
case when each $e$ has cardinality $2$. The rank of a hypergraph $H=(V,E)$
is $\max_{e \in E} |e|$. Recently Kogan and
Krauthgamer \cite{KoganK15} extended Bencz\'{u}r and Karger's sampling
scheme and analysis to show the following: for any weighted hypergraph
$H=(V,E)$ with rank $r$ there is a $(1\pm\eps)$-approximate
cut-sparsifier with $O(n(r + \log n)/\eps^2)$ edges. They show that
sampling with (approximate) strengths will yield the desired sparsifier.
Finding the strengths of edges in hypergraphs can be easily
done in polynomial time. In this note, we develop a near-linear time
algorithm for computing approximate strengths of edges in a weighted
hypergraph. We state a formal theorem and indicate some applications
after we formally define strength of an edge.

\paragraph{Strength of an edge:}
Let $H=(V,E)$ be an \emph{unweighted} hypergraph.  For $U \subseteq
V$, $H[U]=(U,\set{e|e\subset U, e\in E})$ denotes the vertex induced
subhypergraph of $H$. For $A \subset V$ we denote by $\delta_H(A)$ the
set of edges that cross $A$; formally $\delta_H(A) = \{ e \in E \mid e
\cap A \neq \emptyset, e \cap (V \setminus A) \neq \emptyset\}$.  A
hypergraph is $k$-edge-connected if $|\delta_H(A)| \ge k$ for every
non-trivial cut $\emptyset \subsetneq A \subsetneq V$.  The
\EMPH{edge-connectivity} of $H$, denoted by $\lambda(H)$, is the
largest $k$ such that $H$ is $k$-edge-connected. Equivalently,
$\lambda(H)$ is the value of the min-cut in $H$.  The \EMPH{strength}
of an edge $e$, denoted by $\strength_H(e)$, is defined as
$\max_{e\subset U\subset V} \lambda(H[U])$; in other words the largest
connectivty of a vertex induced subhypergraph that contains $e$.  We
also define the cost $\cost_H(e)$ of $e$ as the inverse of the
strength; that is $\cost_H(e) = 1/\strength_H(e)$. We drop the
subscript $H$ if there is no confusion regarding the hypergraph.

The preceding definitions generalize easily to weighted hypergraphs.
Let $w(e)$ be a non-negative integer weight for edge $e$. Then
the notion of edge-connectivity easily generalizes via the weight of
a mincut. The strength of an edge $e$ is the maximum mincut value
over all vertex induced subhypergraphs that contain $e$.

For a hypergraph $H$ we use $n$ to denote the number of vertices,
$m$ to denote the number of edges and $p = \sum_{e \in E(H)} |e|$ to
denote the total degree. We observe that $p$ is the natural representation
size of $H$ when the rank of the hypergraph is not bounded.

Our main technical result is the following.
\begin{theorem}
  \label{thm:main}
  Let $H=(V,E)$ be a edge-weighted hypergraph on $n$ vertices.
  There is an efficient algorithm that computes for each edge
  $e$ an approximate strength $\strength'(e)$ such that the following
  properties are satisfied:
  \begin{enumerate}
  \item lower bound property: $\strength'(e)\leq \strength_H(e)$ and 
  \item $c$-cost property: $\sum_{e\in E} \frac{1}{\strength'(e)} \leq c(n-1)$
    where $c = O(r)$.
  \end{enumerate}
For unweighted hypergraphs the running time of the algorithm is
$O(p \log^2 n)$ and for weighted hypergraphs the running time
is $O(p \log p \log^2 n)$.
\end{theorem}

The function that satisfies the theorem is called a $c$-approximate strength. 
One natural approach is converting the hypergraph to a graph,
and hope the strength approximately carries over. 
For example, replacing each hyperedge with a clique, or a star spanning the vertices in the edge. 
The minimum strength of the replaced edges might give us a $O(c)$-approximation.
Unfortunately, this will not work even when rank is only $3$. 
Consider the following hypergraph $H$ with vertices $\set{v_1,\ldots,v_n}$. 
There is an edge $e=\set{v_1,v_2}$, and edge $\set{v_1,v_2,v_i}$ for all $3\leq i\leq n$. 
The strength of $e$ in $H$ is $1$. 
Let $G$ be a graph where each $\set{v_1,v_2,v_i}$ in $H$ is replaced with a star centered 
at $v_1$ and spans $\set{v_1,v_2,v_i}$. The strength of $e$ in $G$ is $n-1$. 
This bound also holds if each hyperedge $\set{v_1,v_2,v_i}$ is replaced by a clique. 

Our proof of the preceding theorem closely follows the corresponding
theorem for graphs from \cite{BenczurK15}. A key technical tool for graphs
is the deterministic sparsification algorithm for edge-connectivity 
due to Nagamochi and Ibaraki \cite{Nagamochi1992}. Here we rely on
a generalization to hypergraphs from our recent work \cite{ChekuriX17}.

\subsection{Applications}
A weighted hypergraph $H'=(V,E')$ is a $(1\pm \e)$-cut approximation
of a weighted hypergraph $H=(V,E)$, if for every $S\subset V$, 
the cut value of $S$ in $H'$ is
within $(1\pm \e)$ factor of the cut value of $S$ in $H$. 
Below we state formally the sampling theorem from \cite{KoganK15}.

\begin{theorem}[\cite{KoganK15}]
\label{thm:randomsamplebound}
Let $H$ be a rank $r$ weighted hypergraph where edge $e$ has weight
$w(e)$. Let $H_\e$ be a hypergraph obtained by independently sampling
each edge $e$ in $H$ with probability $p_e =
\min \paren{\frac{3((d+2)\ln n + r)}{\strength_H(e)\e^2},1}$ and weight
$w(e)/p_e$ if sampled. With probability at least $1-O(n^{-d})$, the
hypergraph $H_\e$ has $O(n(r+\log n)/\e^2)$ edges and is a $(1\pm
\e)$-cut-approximation of $H$.
\end{theorem}

The expected weight of each edge in $H_\e$ is the weight of the edge
in $H$. It is not difficult to prove that $c$-approximate strength
also suffices to get a $(1\pm \e)$-cut-approximation.  That is, if we
replace $\strength$ with $c$-approximate strength $\strength'$, the
sampling algorithm will still output a $(1\pm \e)$-cut-approximation
of $H$.  Indeed, the lower bound property shows each edge will be
sampled with a higher probability, therefore the probability of
obtaining a $(1 \pm \eps)$-sparsifier increases. However, the number
of edges in the sparsifier increases. The $c$-cost property shows the
hypergraph $H_\e$ will have $O(cn(r+\log n)/\e^2)$ edges.

From \autoref{thm:main}, we can find an $O(r)$-approximate strength 
function $\strength'$ in $O(p\log^2 n\log p)$ time.
\begin{corollary}
  \label{cor:main}
  A $(1\pm \e)$-cut approximation of $H$ with $O(nr(r+\log n)/\e^2)$
  edges can be found in $O(p\log^2 n\log p)$ time with high
  probability.
\end{corollary}
The number of edges in the $(1\pm \e)$-cut approximation is worse than
\autoref{thm:randomsamplebound} by a factor of $r$. It is an
open problem whether the extra factor of $r$ can be removed.

As is the case for graphs, cut sparsification allows for faster
approximation algorithms for cut problems. We mention two below.

\paragraph{Mincut:} The best running time known currently to compute a
(global) mincut in a hypergraph is $O(pn)$ \cite{Klimmek1996,MakW00}. A
$(2+\eps)$-approximation can be computed in $\tilde{O}(p/\eps)$ time
\cite{ChekuriX17}.  Via \autoref{cor:main} we can first sparsify the
hypergraph and then apply the $O(pn)$ time algorithm to the sparsified
graph.  This gives a randomized algorithm that outpus a
$(1+\eps)$-approximate mincut in a rank $r$ hypergraph in $O(n^2r^2(r+
\log n)/\eps^2 + p\log^2 n \log p)$ time and works with high
probability.  For small $r$ the running time is $\tilde{O}(p+n^2/\eps^2)$ for
a $(1+\eps)$-approximation.

\paragraph{$s$-$t$ mincut:} The standard technique to compute a
$s$-$t$ mincut in a hypergraph is computing a $s$-$t$ maximum flow in
an associated capacitated digraph with $O(n+m)$ vertices and $O(p)$
edges \cite{Lawler73}. A $(1-\eps)$ approximation algorithm of $s$-$t$
maximum flow in such graph can be found in $\tilde{O}(p\sqrt{n+m})$
time \cite{LeeS14}.  Via sparsification \autoref{cor:main},
we can obtain a randomized algorithm to find a
$(1+\eps)$-approximate $s$-$t$ mincut in a rank $r$ hypergraph in
$\tilde{O}(n^{3/2}r^2(r+\log n)^{3/2}/\eps^3+p)$ time. For
small $r$ the running time is $\tilde{O}(p+n^{3/2}/\eps^3)$ for a
$(1+\eps)$-approximation.

\section{Preliminaries}
A \EMPH{$k$-strong component} of $H$ is a inclusion-wise maximal set
of vertices $U\subset V$ such that $H[U]$ is $k$-edge-connected.
An edge $e$ is \EMPH{$k$-strong} if $\strength(e) \ge k$, otherwise
$e$ is \EMPH{$k$-weak}.  $\component(H)$ is the number of components
of a hypergraph $H$. The \EMPH{size} of $H$ is $\sum_{e\in E}
|e|$. The \EMPH{degree} of a vertex $v$, $\deg(v)$, is the number of
edges incident to $v$.  $n$ denotes the number of vertices in $H$, $p$
denotes the size of $H$ and $r$ denotes the rank of $H$.

Contracting an edge $e$ in a hypergraph $H=(V,E)$ results in a new
hypergraph $H'=(V',E')$  where all vertices in $e$ are identified into
a single new vertex $v_e$. Any edge $e' \subseteq e$
is removed. Any edge $e'$ that properly intersects with $e$ is adjusted by
replacing $e' \cap e$ by the new vertex $v_e$. We note $H'$ by
$H/e$.

\begin{proposition}
\label{pro:fundamental}
Deleting an edge does not increase the strength of any remaining edge. 
Contracting an edge does not decrease the strength of any remaining edge.
\end{proposition}
\begin{lemma}
\label{lem:strengthboundbycut}
Let $H=(V,E)$ be a hypergraph. If an edge $e\in E$ crosses a cut of
value $k$, then $\strength_H(e)\leq k$.
\end{lemma}
\begin{proof}
  Suppose $S$ is a cut such that $|\delta_H(S)| \le k$ and $e \in \delta_H(S)$. 
  Consider any $U\subset V$ that contains $e$ as a subset and let
  $H' = H[U]$.  It is easy to see that $|\delta_{H'}(S \cap U)| \le k$
  and hence $\lambda(H') \le k$. Therefore, $\strength(e) =
  \max_{e\subset U\subset V} \lambda(H[U]) \leq k$.
\end{proof}

\begin{lemma}[Extends Lemma 4.5,4.6 \cite{BenczurK15}]
\label{lem:deletecontract}
Let $e,e'$ be distinct edges in a hypergraph $H=(V,E)$.
\begin{enumerate}
\item If $\strength_H(e)\geq k$ and $e'$ is a $k$-weak edge
then $\strength_H(e) = \strength_{H'}(e)$ where $H' = H \setminus e'$.
\item Suppose $\strength_H(e)< k$ (that is $e$ is a $k$-weak edge) and
  $e'$ is a $k$-strong edge. Let $H' = H/e'$ obtained by contracting $e'$.
  If $f$ is the corresponding edge of $e$ in $H'$ then
  $\strength_H(e) = \strength_{H'}(f)$.
\end{enumerate}
In short, contracting a $k$-strong edge does not increase the strength of a
$k$-weak edge and deleting a $k$-weak edge does not decrease the strength of
a $k$-strong edge.
\end{lemma}
\begin{proof} We prove the two claims separately.

  \subparagraph{Deleting a $k$-weak edge:} Since $e$ is $k$-strong in $H$
  there is $U \subseteq V$ such that $\lambda(H[U]) \ge k$ and
  $e \subseteq U$. If $e' \subseteq U$, $e'$ would be a $k$-strong
  edge. Since $e'$ is $k$-weak, $e'$ does not belong to $H[U]$. 
  Thus $H'[U] = H[U]$ which implies that $\strength_{H'}(e) \ge k$.

  \subparagraph{Contracting a $k$-strong edge:} Let $H' = H/e'$ where
  $e'$ is a $k$-strong edge in $H$. Let $v_{e'}$ be the vertex in $H'$
  obtained by contracting $e'$. Let $U' \subseteq V(H')$ be the set
  that certifies the strength of $f$, that is $\strength_{H'}(f) =
  \lambda(H'[U'])$ and $f \subseteq U'$. If $v_{e'} \not \in U'$ then
  it is easy to see that $U' \subseteq V(H)$ and $H[U] = H'[U']$ which
  would imply that $\strength_{H'}(f) = \strength_H(e)$. Thus, we can
  assume that $v_{e'} \in U'$. Let $U$ be the set of vertices
  in $V(H)$ obtained by uncontracting $v_{e'}$. We claim that
  $\lambda(H[U]) \ge \lambda(H'[U'])$. If this is the case we would
  have $\strength_H(e) \ge \lambda(H[U]) \ge \lambda(H'[U']) = \strength_{H'}(f)$.  We now
  prove the claim. Let $W$ be the set that certifies the strength of $e'$, namely
  $\strength_{H}(e') = \lambda(H[W])$. Consider any cut $S\subset U\cup W$ in
  $H[U\cup W]$. If $S$ crosses $W$ or strictly contained in $W$,
  then $S$ has cut value at least $\lambda(H[W])$. Otherwise, $S\subset U\setminus W$ or 
  $W\subset S$. By symmetry, we only consider $S\subset U\setminus W$. 
  $S\neq U$ because $e'\subset U\cap W$. $S$ is also a 
  cut in $H'[U']$, and $|\delta_{H[U]}(S)| = |\delta_{H'[U']}(S)|$. 
  This shows that 
  \[
    \lambda(H[U\cup W]) \geq \min(\lambda(H'[U']),\lambda(H[W])) 
  \]
  Because $\lambda(H[W])>\lambda(H[U])$, and $\lambda(H[U])\geq \lambda(H[U\cup W])$,
  we arrive at $\lambda(H[U])\geq \lambda(H'[U'])$.
 \end{proof}

\begin{theorem}[Extends Lemma 4.10 \cite{BenczurK15}]
\label{thm:cutval1}
Consider a connected unweighted hypergraph $H=(V,E)$.  
A weighted hypergraph $H'=(V,E)$
with weight $\cost_H(e)$ on each edge $e$ has minimum cut value $1$.
\end{theorem}
\begin{proof}
  Consider a cut $S$ of value $k$ in $H$; that is $|\delta_H(S)|= k$.
  For each edge $e\in \delta(S)$, $\strength_H(e) \leq k$ by
  \autoref{lem:strengthboundbycut}. Therefore $e$ has weight at least
  $1/k$ in $H'$. It follows the cut value of $S$ in $H'$ is
  at least $k \cdot 1/k \ge 1$. Thus, the mincut of $H'$ is at least $1$.

  Let $S$ be a mincut in $H$ whose value is $k^*$. It is easy to see
  that for each $e \in \delta_H(S)$, $\strength_H(e) = k^*$.
  Thus the value of the cut $S$ in $H'$ is exactly $k^* \cdot 1/k^* = 1$.
\end{proof}

\begin{lemma}[Extends Lemma 4.11 \cite{BenczurK15}]
\label{lem:strength1cost}
For a unweighted hypergraph $H=(V,E)$ with at least $1$ vertex,  
\[
\sum_{e\in E} \cost(e) \leq n-\component(H)
\]
\end{lemma}
\begin{proof}
Let $t = n-\component(H)$. We prove the theorem by induction on $t$.

For the base case, if $t=0$, then $H$ has no edges and therefore the
sum is $0$.

Otherwise, let $t>0$. Let $U$ be a connected component of $H$ with at
least $2$ vertices. There exists a cut $X$ of cost $1$ in $H[U]$ by
\autoref{thm:cutval1}.

Let $H'=H-\delta(S)$.  $H'$ has at least one more connected component
than $H$.  By \autoref{pro:fundamental}, for each $e \in E(H')$
$\strength_{H}(e) \ge \strength_{H'}(e)$; hence $\cost_{H'}(e) \ge
\cost_H(e)$.  By the inductive hypothesis, $\sum_{e \in E(H')}
\cost_{H'}(e) \le (n - \component(H')) \le t-1$. The edges in $H$ are
exactly $\delta(X)\cup E(H')$. By the same argument as in the
preceding lemma, $\sum_{e \in \delta(X)} \cost_H(e) = 1$.  Therefore,
\begin{align*}
\sum_{e \in E(H)} \cost_H(e) & =  \sum_{e \in \delta(X)} \cost_H(e) + \sum_{e \in E(H')} \cost_H(e) \\
 & =  1 + \sum_{e \in E(H')} \cost_H(e) \\
 & \le  1 + \sum_{e \in E(H')} \cost_{H'}(e) \\
 & \le  1 + (t-1) = t.
\end{align*}
\end{proof}

\begin{corollary}
\label{cor:weakedgebound}
The number of $k$-weak edges in an unweighted hypergraph $H$ on $n$ vertices is at most $k(n-\component(H))$. 
\end{corollary}

\begin{lemma}
\label{lem:strongcomponentdisjoint}
The $k$-strong components are pairwise disjoint.  
\end{lemma}
\begin{proof}
Consider $k$-strong components $A$ and $B$. Assume $A\cap B\neq \emptyset$, then $\lambda(G[A\cup B])\geq \min(\lambda(G[A]),\lambda(G[B])) \geq k$ using triangle inequality for connectivity. This shows $A= A\cup B=B$ by maximality of $A$ and $B$.
\end{proof}
\section{Estimating strengths in unweighted hypergraphs}
\label{sec:unweighted}
In this section, we consider unweighted hypergraphs and describe a
near-linear time algorithm to estimate the strengths of all edges as
stated in \autoref{thm:main}. Let $H=(V,E)$ be the given hypergraph.
The high-level idea is simple. We assume that there is a fast
algorithm $\textsc{WeakEdges}(H,k)$ that returns a set of edges
$E'\subset E$ such that $E'$ contains all the $k$-weak edge in $E$;
the important aspect here is that the output may contain some edges
which are \emph{not} $k$-weak, however, the algorithm should not output
too many such edges (this will be quantified later).

The estimation algorithm is defined in \autoref{alg:estimation}. 
The algorithm repeatedly calls $\textsc{WeakEdges}(H,k)$
for increasing values of $k$ while removing the
edges found in previous iterations. 


\begin{figure}
\centering
\begin{algorithm}
\textsc{Estimation}($H$)\+
\\  Compute a number $k$ such that $\strength_H(e) \ge k$ for all $e \in E(H)$
\\  $H_0 \gets H$, $i\gets 1$
\\  while there are edges in $H_{i-1}$\+
\\    $F_i \gets$ \textsc{WeakEdges}$(H_{i-1},2^i k)$
\\    for $e\in F_i$\+
\\      $\strength'(e)\gets 2^{i-1} k$\-
\\    $H_i\gets H_{i-1}-F_i$
\\    $i \gets i+1$\-
\\  return $\strength'$
\end{algorithm}
\caption{The estimation algorithm}
\label{alg:estimation}
\end{figure}

\begin{lemma}
\label{lem:ei}
Let $H=(V,E)$ be an unweighted hypergraph. Then,
\begin{enumerate}
\item For each $e \in F_i$, $\strength(e) \ge 2^{i-1} k$. That is, the strength
of all edges deleted in iteration $i$ is at least $2^{i-1}k$.
\item For each $e \in E$, $\strength'(e) \le \strength(e)$.
\end{enumerate}
\end{lemma}
\begin{proof}
  $\strength_{H_i}(e)\leq \strength_{H}(e)$ for all $i$ and $e\in
  E(H_i)$ because deleting edges cannot increase strength. Let $E_i$ denote
  the set of edges in $H_i$ with $E_0 = E$.

  We prove that $2^{i}k\leq \strength_{H_{i}}(e)$ for all $e\in E_i$
  by induction on $i$.  If $i=0$, then $k \leq \strength_{H_0}(e)$ for
  all $e\in E_0$.  Now we assume $i>0$.  At end of
  iteration $(i-1)$, by induction, we have that
  $\strength_{H_{i-1}}(e) \ge 2^{i-1}k$ for each $e \in E_{i-1}$.  In
  iteration $i$, $F_i$ contains all $2^{i}k$-weak edges in the 
  graph $H_{i-1}$. We hav $E_i = E_{i-1} - F_i$. Thus, for any
  edge $e \in E_i$, $\strength_{H}(e) \ge \strength_{H_{i-1}}(e) \ge
  2^{i}k$. This proves the claim.

  Since $F_i \subset E_{i-1}$, it follows from the previous claim that
  $\strength_{H}(e) \ge 2^{i-1}k$ for all $e \in F_i$. 
  Since the $F_i$ form a partition of $E$, we have
  $\strength'(e)\leq \strength_H(e)$ for all $e\in E$.
\end{proof}

Note that in principle \textsc{WeakEdges}$(H,k)$ could output all
the edges of the graph for any $k \ge \lambda(H)$. This would result
in a high cost for the resulting strength estimate.  Thus, we need
some additional properties on the output of the procedure
\textsc{WeakEdges}$(H,k)$.  Let $H=(V,E)$ be a hypergraph.  A set
of edges $E'\subset E$ is called \EMPH{$\ell$-light} if $|E'|\leq \ell
(\component(H-E')-\component(H))$.  Intuitively, on average, we remove
$\ell$ edges in $E'$ to increase the number of components of $H$ by $1$.

\begin{lemma}
\label{lem:ei2}
If $\textsc{WeakEdges}(H,k)$ outputs a set of $ck$-light edges for all 
$k$, then the output $\strength'$ of the algorithm
$\textsc{Estimation}(H)$ satisfies the $2c$-cost property.
That is, $\sum_{e \in E} \frac{1}{\strength'(e)} \le 2c(n-1)$.
\end{lemma}
\begin{proof}
  From the description of $\textsc{Estimation}(H)$, and using the
  fact that the edges sets $F_1,F_2,\ldots,$ partition $E$, we have 
  $\sum_{e \in E} \frac{1}{\strength'(e)} = \sum_{i \ge 1} |F_i|
  \frac{1}{2^{i-1}k}$.

$F_i$ is the output of $\textsc{WeakEdges}(H_{i-1},2^ik)$.
From the lightness property we assumed, $|F_i| \le c 2^ik (\component(H_i) - \component(H_{i-1}))$. Combining this with the preceding equality, 
\[
\sum_{e \in E} \frac{1}{\strength'(e)} = \sum_{i \ge 1} |F_i|
  \frac{1}{2^{i-1}k} \le \sum_{i \ge 1} 2c (\component(H_i) - \component(H_{i-1})) \leq 2c (n-1).
  \]
\end{proof}

\subsection{Implementing \textsc{WeakEdges}}
We now show describe an implementation of \textsc{WeakEdges}$(H,k)$ 
that outputs a $4rk$-light set.

Let $H=(V,E)$ be a hypergraph.  An edge $e$ is \EMPH{$k$-crisp} with
respect to $H$ if it crosses a cut of value less than $k$. In other
words, there is a cut $X$, such that $e\in \delta(X)$ and
$|\delta(X)|< k$. Note that any $k$-crisp edge is $k$-weak.  A set of
edges $E'\subset E$ is a \EMPH{$k$-partition}, if $E'$ contains all
the $k$-crisp edges in $H$.  A $k$-partition may contain non-$k$-crisp
edges.

We will assume access to a subroutine $\textsc{Partition}(H,k)$ that
given $H$ and integer $k$, it finds a $2k$-light $k$-partition of $H$. We will
show how to implement \textsc{Partition} later. See
\autoref{alg:weakedge} for the implementation of \textsc{WeakEdges}.

\begin{figure}[htb]
\centering
\begin{algorithm}
\textsc{WeakEdges}$(H,k)$\+
\\ $E'\gets \emptyset$
\\ repeat $1+\log_2 n$ times:\+
\\   $E'\gets E'\cup \textsc{Partition}(H,2rk)$
\\   $H\gets H-E'$\-
\\ return $E'$\-
\end{algorithm}
\caption{Algorithm for returning a $4rk$-light set of all $k$-weak edges in $H$.}
\label{alg:weakedge}
\end{figure}
\begin{theorem}
\label{thm:weakedgerun}
\textsc{WeakEdges}$(H,k)$ returns a $4rk$-light set $E'$ such that $E'$ contains all the $k$-weak edges of $H$ with $O(\log n)$ calls to $\textsc{Partition}$.
\end{theorem}
\begin{proof}
  First, we assume $H$ has no $k$-strong component with more than $1$
  vertex and that $H$ is connected. Then all edges are $k$-weak and
  the number of $k$-weak edges in $H$ is at most $k(n-1)$ by
  \autoref{cor:weakedgebound}. It also implies that $\sum_{v} \deg(v)
  \leq rk(n-1)$.  By Markov's inequality at least half the vertices have degree less than $2rk$.  For any vertex $v$ with degree less
  than $2rk$, all edges incident to it are $2rk$-crisp.  Thus, after
  the first iteration, all such vertices become isolated since
  $\textsc{Partition}(H,2rk)$ contains all $2rk$-crisp edges. If $H$ is not connected then we can
  apply this same argument to each connected component and deduce that
  at least half of the vertices in each component will be isolated in
  the first iteration Therefore, in
  $\log n$ iterations, all vertices become isolated.  Hence
  $\textsc{WeakEdges}(H,k)$ returns all the edges of $H$.
  
  Now consider the general case when $H$ may have $k$-strong
  components.  We can apply the same argument as above to the
  hypergraph obtained by contracting each $k$-strong component into
  a single vertex. The is well defined because the $k$-strong components
  are disjoint by \autoref{lem:strongcomponentdisjoint}.
  
  Let the edges removed in the $i$th iteration to be $E_i$, and the hypergraph
  before the edge removal to be $H_i$. So $H_{i+1}=H_i-E_i$.  
  Recall that \textsc{Partition}$(H_i,2rk)$ returns a $4rk$-light set. Hence
  we know $|E_i| \leq 4rk(\component(H_{i+1})-\component(H_i))$.

  \[
  |E'| = \sum_{i\geq 1} |E_i| \leq \sum_{i\geq 1} 4rk(\component(H_{i+1})-\component(H_i)) = 4rk(\component(H-E')-\component(H))
  \]
  This shows $E'$ is $4rk$-light.
\end{proof}

It remains to implement \textsc{Partition}($H,k$) that returns a
$2k$-light $k$-partition. To do this, we introduce $k$-sparse
certificates.  Let $H=(V,E)$ be a hypergraph.  A subset of edges
$E'\subset E$, define $H'=(V,E')$ is a \EMPH{$k$-sparse certificate}
if $|\delta_{H'}(A)|\geq \min(|\delta_{H}(A)|,k)$ for all $A\subset V$
where $H'=(V,E')$. 

\begin{theorem}
\label{thm:certificateruntime}
There is an algorithm $\textsc{Certificate}(H,k)$ that given hypergraph $H$
and integer $k$, finds a $k$-sparse certificate $E'$ of $H$ in $O(p)$ time
such that $|E'|\leq k(n-1)$.
\end{theorem}
\begin{proof}
See \autoref{sec:sparsecertificate}.
\end{proof}

A $k$-sparse certificate $E'$ is certainly a $k$-partition.  However,
$E'$ may contain too many edges to be $2k$-light. Thus, we would like
to find a smaller subset of $E'$. Note that every $k$-crisp edge 
must be in a $k$-sparse certificate and hence no edge in
$E\setminus E'$ can be $k$-crisp.  Hence we will
contract the edges in $E\setminus E'$, and find a $k$-sparse
certificate in the hypergraph after the contraction. We repeat the
process until eventually we reach a $2k$-light set.  See
\autoref{alg:partition} for the formal description of the algorithm.

\begin{figure}
\centering
\begin{algorithm}
\textsc{Partition}($H,k$)\+
\\if number of edges in $H\leq 2k(n-\component(H))$\+
\\  $E'\gets $edges in $H$
\\  return $E'$\-
\\  else\+
\\  $E'\gets $ $\textsc{Certificate}(H,k)$
\\  $H\gets $ contract all edges of $H-E'$
\\  return \textsc{Partition}($H,k$)\-
\end{algorithm}
\caption{Algorithm for returning a $2k$-light $k$-partition.}
\label{alg:partition}
\end{figure}

\begin{theorem}
\label{thm:partitionruntime}
\textsc{Partition}($H,k$) outputs a $2k$-light $k$-partition in $O(p\log n)$ time.  
\end{theorem}
\begin{proof}
  If the algorithm either returns all the edges of the graph $H$ 
  in the first step then it is easy to see that the output is
  a $2k$-light $k$-partition since the algorithm explicitly checks
  for the lightness condition.

  Otherwise let $E'$ be the output of $\textsc{Certificate}(H,k)$.
  As we argued earlier, $E - E'$ contains no $k$-crisp edges and hence 
  contracting them is safe. Moreover, all the original $k$-crisp edges
  remain $k$-crisp after the contraction. Since the algorithm recurses
  on the new graph, this establishes the correctness of the output.

  We now argue for termination and running time by showing that the
  number of vertices halves in each recursive call.  Assume $H$
  contains $n$ vertices, the algorithm finds a $k$-sparse certificate
  and contracts all edges not in the certificate.  The resulting
  hypergraph has $n'$ vertices and $m'$ edges. We have $m'\leq k(n-1)$
  by \autoref{thm:certificateruntime}.  If $n'-1\leq (n-1)/2$, then
  the number of vertices halved. Otherwise $n'-1>(n-1)/2$, then the
  number of edges $m'\leq k(n-1)<2k(n'-1)$, and the algorithm
  terminates in the next recursive call.

  The running time of the algorithm for a size $p$ hypergraph with $n$
  vertices is $T(p,n)$. $T(p,n)$ satisfies the recurrence $T(p,n) =
  O(p) + T(p,n/2) = O(p\log n)$.
\end{proof}

Putting things together, $\textsc{Estimation}(H)$ finds
the desired $O(r)$-approximate strength. 
\begin{theorem}
\label{thm:estimationtime}
Let $H=(V,E)$ be a unweighted hypergraph. 
The output $\strength'$ of $\textsc{Estimation}(H)$ is a $O(r)$-approximate 
strength function of $H$.
$\textsc{Estimation}(H)$ can be implemented in $O(p\log^2 n \log p)$ time.
\end{theorem}
\begin{proof}
  Combining \autoref{lem:ei}, \autoref{lem:ei2} and \autoref{thm:weakedgerun}, we 
  get the output $\strength'$ of $\textsc{Estimation}(H)$ is a $O(r)$-approximate 
  strength function.
  The maximum strength in the graph is at most $p$, all edges with be removed
  at the $(1+\log p)$th iteration of the while loop. 
  In each iteration, there is one call to \textsc{WeakEdges}. 
  Each call of \textsc{WeakEdges} takes $O(p\log^2 n)$ time by combining
  \autoref{thm:partitionruntime} and \autoref{thm:weakedgerun}.
  The step outside the while loop takes linear time, since we can set $k$ to be $1$
  as a lower bound of the strength.
  Hence overall, the running time is $O(p\log^2 n \log p)$.
\end{proof}

\section{Estimating strengths in weighted hypergraphs}

Consider a weighted hypergraph $H=(V,E)$ with an associated weight
function $w:E\to \N_+$; that is, we assume all weights are
non-negative integers. 
For proving correctness, we consider an
unweighted hypergraph $H'$ that simulates $H$. Let $H'$ contain $w(e)$
copies of edge $e$ for every edge $e$ in $H$; one can see that the
strength of each of the copies of $e$ in $H'$ is the same as the
strength of $e$ in $H$. Thus, it suffices to compute strengths of edges
in $H'$. We can apply the correctness proofs from the previous
sections to $H'$.  In the remainder of the section we will only be
concerned with the running time issue since we do not wish to
explicitly create $H'$. We say $H$ is the implicit representation of $H'$.

We can implement $\textsc{Certificate}(H,k)$ that finds an implicit representation
of the $k$-sparse certificate $E'$ of $H'$ such that $|E'|\leq k(n-1)$, in $O(p+n\log n)$ time,
where $p$ is the number of edges in the implicit representation, see \autoref{sec:sparsecertificate}. 

The remaining operations in \textsc{Partition} and \textsc{WeakEdges} consist
only of adding edges, deleting edges and contracting edges. These 
operations take time only depending on the size of the implicit representation.
Therefore the running time in \autoref{thm:partitionruntime} and \autoref{thm:weakedgerun}
still holds for weighted hypergraph. 

\begin{theorem}
\label{thm:estimationtimeweight}
Given $b$ a lower bound of the strength. If the total weight
of all edges is at most $bM$, then $\textsc{Estimation}(H)$ can be
implemented in $O(p\log^2 n \log M)$ time.
\end{theorem}
\begin{proof}
  Because $b$ is a lower bound of the strength, we can set $k$ to be $b$ in the
  first step.
  The maximum strength in the graph is at most $bM$, all edges will be removed
  at the $(1+\log M)$th iteration of the while loop. 
  Each iteration calls $\textsc{WeakEdges}$ once, hence the running time is
  $O(p\log^2 n \log M)$.
\end{proof}

The running time in \autoref{thm:estimationtimeweight} can be improved to strongly polynomial time by using the windowing technique \cite{BenczurK15}.

Assume we have disjoint intervals $I_1,\ldots,I_t$ where for every
$e\in E$, $\strength(e)\in I_i$.  In addition, assume $I_i=[a_i,b_i]$,
$b_i \leq p^2 a_i$ and $b_i\leq a_{i+1}$ for all $i$.  We can
essentially apply the estimation algorithm to edges with strength
inside each interval. Let $E_i$ be the set of edges whose strengh lies
in interval $I_i$.  Indeed, let $H_i$ to be the graph obtained from
$H$ by contracting all edges in $E_j$ where $j>i$, and deleting all
edges in $E_{j'}$ where $j'<i$. For edge $e \in E_i$ let $e'$ be its
corresponding edge in $H_i$.  From \autoref{lem:deletecontract},
$\strength_{H_i}(e') = \strength_H(e)$.  The total weight of $H_i$ is
at most $p b_i\leq p^3 a_i$.  We can run $\textsc{Estimation}(H_i)$ to
estimate $\strength_{H_i}$ since the ratio between the lower bound
$a_i$ and upper bound $b_i$ is $p^3$.  Let $p_i$ be the size of $H_i$,
and $n_i$ be the number of vertices in $H_i$, the running time for
$\textsc{Estimation}(H_i)$ is $O(p_i\log^2 n_i\log
p_i)$ by \autoref{thm:estimationtimeweight}.  The total running time of
\textsc{Estimation} over all $H_i$ is $O(\sum_{i} p_i \log^2 n_i \log p_i)
= O(p \log^2 n \log p)$.  Constructing $H_i$ from $H_{i+1}$ takes
$O(p_i+p_{i+1})$ time: contract all edges in $H_{i+1}$ and then add
all the edges $e_i$ where $\strength(e)\in I_i$.  Therefore we can
construct all $H_1,\ldots,H_t$ in $O(p)$ time.

It remains to find the intervals $I_1,\ldots,I_t$.  For each edge $e$,
we first find values $d_e$, such that $d_e \leq \strength(e)\leq p
d_e$.  The maximal intervals in $\bigcup_{e\in E}[d_e,p d_e]$ are the
desired intervals. We now describe the procedure to find the values
$d_e$ for all $e \in E$.

\begin{definition}
  The star approximate graph $A(H)$ of $H$ is a weighted graph
  obtained by replacing each hyperedge $e$ in $H$ with a star $S_e$,
  where the center of the star is an arbitrary vertex in $e$, and the
  star spans each vertex in $e$. Every edge in $S_e$ has weight equal
  to the weight of $e$.
\end{definition}

It is important that $A(H)$ is a multigraph: parallel edges are
distinguished by which hyperedge it came from.  We define a
correspondence between the edges in $A(H)$ and $H$ by a function
$\pi$. For an edge $e'$ in $A(H)$, $\pi(e')=e$ if $e'\in S_e$.  Let
$T$ be a maximum weight spanning tree in $A(H)$.  For $e\in E$, define $T_e$
to be the minimal subtree of $T$ that contains all vertices in $e$.
Note that all the leaves of $T_e$ are vertices from $e$.

For any two vertices $u$ and $v$, we define $d_{uv}$ to be the weight
of the minimum weight edge in the unique $u$-$v$-path in $T$.  For
each edge $e \in E$ we let $d_e= \min_{u,v\in e} d_{uv}$. We will show
$d_e$ satisfies the property that $d_e\leq \strength(e)\leq pd_e$.

Let $V_e = \bigcup_{e'\in T_e} \pi(e')$. Certainly, $\strength(e)\geq
\lambda(H[V_e])$, because all vertices of $e$ are contained in $V_e$.
$\lambda(H[V_e])\geq d_e$ because every cut in $H[V_e]$ has to cross
some $\pi(e')$ where $e'\in T_e$, and the weight of $\pi(e')$ is at
least $d_e$.  Hence $\strength(e)\geq d_e$.

We claim if we remove all edges in $H$ with weight at most $d_e$, then
it will disconnect some $s,t\in e$. If the claim is true then $e$
crosses a cut of value at most $p d_e$.  By
\autoref{lem:strengthboundbycut}, $\strength(e) \leq p d_e$.  Assume
that the claim is not true. Then, in the graph $A(H)$ we can remove
all edges with weight at most $d_e$ and the vertices in $e$ will still
be connected. We can assume without loss of generality that the
maximum weight spanning tree $T$ in $A(H)$ is computed using the
greedy Kruskal's algorithm. This implies that $T_e$ will contain only
edges with weight strictly greater than $d_e$. This contradicts the
definition of $d_e$.

$A(H)$ can be constructed in $O(p)$ time.  The maximum spanning tree
$T$ can be found in $O(p+n\log n)$ time.  We can construct a data
structure on $T$ in $O(n)$ time, such that for any $u,v\in V$, it
returns $d_{uv}$ in $O(1)$ time. \cite{Chazelle1987}
To compute $d_e$, we fix some vertex
$v$ in $e$, and compute $d_e = \min_{u\in e,v\neq u} d_{uv}$ using the
data structure in $O(|e|)$ time.  Computing $d_e$ for all $e$ takes in
$O(\sum_{e\in E} |e|) = O(p)$ time. The total running time is
$O(p+n\log n)$.  We conclude the following theorem.

\begin{lemma}
  Given a weighted hypergraph $H$, we can find a value $d_e$ for each
  edge $e$, such that $d_e\leq \strength(e)\leq p d_e$ in $O(p+n\log
  n)$ time.
\end{lemma}

The preceding lemma gives us the desired intervals. 
Using \autoref{thm:estimationtimeweight}, we have the desired theorem. 
\begin{theorem}
\label{thm:8rapproximation}
Given a rank $r$ weighted hypergraph $H$ with weight function $w$, in
$O(p\log^2 n\log p)$ time, one can find a $O(r)$-approximate strength
function of $H$.
\end{theorem}

\bibliographystyle{plain}
\bibliography{hypergraph}

\appendix
\section{$k$-sparse certificate}
\label{sec:sparsecertificate}
We show how to find a $k$-sparse certificate for weighted and unweighted
hypergraphs.

We refer to \cite{ChekuriX17} for terminology. Given a hypergraph $H=(V,E)$ consider an MA-ordering of vertices $v_1,\ldots,v_n$. 
Let $e_1,\ldots,e_m$ be the induced head ordering of the edges. 

For the unweighted version, let $D_k(v)$ be the first $k$ backward edges of $v$ in the
head ordering, or all the backward edges of $v$ if there are fewer
than $k$ backward edges. 
Given $H$ and $k$, $H_k=(V,E_k)$ is defined such that $E_k=\bigcup_{v} D_k(v)$.
It is easy to see $H_k$ can be constructed in linear time once the MA-ordering
have been computed. It is a $k$-sparse certificate follows directly from \cite{ChekuriX17}.
It is also obvious it contain at most $k(n-1)$ edges, since $|D_k(v)|\leq k$ and $D_k(v_1)=\emptyset$. 

The algorithm is similar for the weighted version. For a vertex $v$, assume its
backward edges are $e_1,\ldots,e_\ell$ with weights $w_1,\ldots,w_\ell$, and if $i<j$ then
$e_i$ comes before $e_j$ in the head order of the edges.
Define $w_v(e_i) = \max(\min(k-\sum_{j<i} w_j, w_i),0)$ and $0$ for all other edges.
Define $w'(e) = \max_{v\in V} w_v(e)$.
A weighted hypergraph with the weight function $w'$ is a $k$-sparse certificate of $H$.
With careful bookkeeping, $w'$ can be computed in $O(p)$ time.

The running time is dominated by finding the MA-ordering, which is $O(p+n\log n)$ time. 
\end{document}